\documentclass[11pt]{article}

\usepackage{fullpage}

\usepackage{amsfonts}
\usepackage{graphicx}
\usepackage{amsmath}
\usepackage[english]{babel}
\usepackage{amsthm}
\usepackage{verbatim}
\usepackage{enumerate}
\usepackage{tikz}
\usetikzlibrary{arrows,decorations.pathmorphing,backgrounds,positioning,fit,automata}
\usetikzlibrary{petri}
\usetikzlibrary{topaths}
\usetikzlibrary{snakes}
\usetikzlibrary{shapes.misc}
\usepackage{float}
\usepackage{authblk}

\newtheorem{theorem}{Theorem}[section]

\newtheorem{corollary}[theorem]{Corollary}

\newtheorem{lemma}[theorem]{Lemma}
\newtheorem{remark}[theorem]{Remark}

\newtheorem{proposition}[theorem]{Proposition}

\theoremstyle{definition}
\newtheorem{example}[theorem]{Example}

\usepackage{todonotes}

\providecommand{\keywords}[1]{\textbf{Keywords.} #1}

\title{Negative Prices in Network Pricing Games}
\author[1]{Andr\'{e}s~Cristi}
\author[2]{Marc~Schr\"oder}

\affil[1]{Universidad de Chile, Santiago, Chile, andres.cristi@ing.uchile.cl}
\affil[2]{RWTH Aachen University, Aachen, Germany, marc.schroeder@oms.rwth-aachen.de}

\date{}

\begin{document}
\maketitle

\begin{abstract}
	In a Stackelberg network pricing game, a leader sets prices for a given subset of edges so as to maximize profit, after which one or multiple followers choose a shortest path from their source to sink. We study the counter-intuitive phenomenon that the use of negative prices by the leader may in fact increase his profit. In doing so, we answer the following two questions. First, how much more profit can the leader earn by setting negative prices? Second, for which network topologies can the profit be increased by using negative prices? Our main result shows that the profit with negative prices can be a factor $\Theta(\log (m\cdot\bar k))$ larger than the maximum profit with positive prices, where $m$ is the number of priceable edges in the graph and $\bar k \leq 2^m$ the number of followers. In particular, this factor cannot be bounded for the single follower case, and can even grow linearly in $m$ if the number of followers is large. Our second result shows that series-parallel graphs with a single follower and Stackelberg games with matroid followers are immune to the negative price paradox.
\end{abstract}


\section{Introduction}

Suppose that you are planning to buy an intercontinental flight. It is not uncommon to find that a direct flight between two main hubs of a carrier is more expensive than taking an extra leg to a smaller hub. So adding an extra flight to your itinerary, might in fact decrease your costs. One way of explaining this behavior of flight carriers is the following. Suppose there are two classes of customers: businessmen who are willing to pay a high price for the direct flight, and tourists who are willing to pay a lower price, but need to travel longer to reach their destination. In order to extract both types' surplus, it is optimal to set a higher price for the direct flight and a lower price for the longer connecting flight.

The lower price for a connecting flight is an example of bundle pricing: the price of two connecting flights is cheaper than buying both flights separately. In other words, firms are willing to put a negative price on a good so as to sell multiple goods instead of only one. The example above shows that this can be beneficial when having multiple customers, but a similar result applies when having just a single customer. With a single customer however, the question whether bundle pricing is better than single item pricing depends on the set of feasible bundles of the customer. Whenever goods are either perfect complements or perfect substitutes, bundle pricing will not improve your profit. The case where customers buy paths from an origin to a destination is of special interest, as it is a simplified model of many transportation markets.

We consider the class of Stackelberg network pricing games to model the above situation. These models are typically used for road tolling problems and were first introduced by Labb\'{e} et al.~\cite{LaMaSa98}. In Stackelberg network pricing games, a leader moves first by setting prices on edges he owns, after which each follower decides on a path of minimum cost between her source and sink. The objective of the leader is to maximize profit. It is therefore a natural and common assumption in the literature to assume that prices are non-negative, see, e.g., \cite{BiGuPrWi08,BrChKhLaNa10,BrHoKr12,Jo11,Ho08}. However, Labb\'{e} et al.~\cite{LaMaSa98} gave an example of a Stackelberg single-follower shortest-path pricing game in which the profit is maximized by using negative prices. We call this phenomenon the \textit{negative price paradox}. The use of negative prices in that example can be thought of as a bundle pricing strategy of the leader that guarantees that the follower uses multiple edges owned by the leader. The main question we want to answer is how much more profit can the leader earn by using such bundle pricing strategies. Since calculating optimal bundle prices is rather intractable as a price has to be determined for each of the exponentially many bundles of resources, we restrict our analysis to single item pricing. Thus, we model bundle pricing by allowing for negative prices of the leader.

For multiple followers, it might not be surprising that the profit can be arbitrarily larger with negative prices compared to restricting to positive prices. However, our main result shows that the same is true for a single follower. So even though the assumption of non-negative prices is natural, it clearly has an impact on the maximum profit the leader can achieve. Our contribution is two-fold. From a practical perspective, we show that a very simple bundle pricing mechanism, using positive and negative single item prices, can achieve arbitrarily higher profits than single item pricing with only positive prices. From a theoretical perspective, we show that a seemingly innocent assumption, non-negative prices, might have a big impact on the outcome of the game.

\subsection{Contribution}
We start by studying Stackelberg network pricing games in which the followers choose a shortest path from source to sink. Our main goal is to quantify the loss in profit due to assuming non-negative prices. For this purpose, we define the \textit{Price of Positivity} (PoP), which is the ratio between the profit of the leader when he is allowed to use negative prices and when he is not. Theorem \ref{thm:main} proves that the PoP can be of order $\Theta(\log m\cdot\bar k)$, where $m$ is the number of priceable edges and $\bar k \leq 2^m$ the number of followers. We prove that this bound is asymptotically tight by means of two different classes of instances. First, Theorem \ref{thm:inf} shows that the price of positivity can be arbitrarily large even with a single follower. To prove this, we use the class of generalized Braess graphs \cite{Ro06}. Second, Theorem \ref{thm:mul} shows that the price of positivity can be arbitrarily large in path graphs given that there are sufficiently many followers.

Then we turn to the question of which network topologies are immune to the negative price paradox. A network is immune to the paradox if for all instances within that network the negative price paradox cannot occur. Theorem \ref{thm:sur} proves that in series-parallel graphs with a single follower, the leader can extract all surplus from the follower by means of positive prices. We then show that in fact, with a single follower, series-parallel graphs are exactly the class of networks that are immune to the negative price paradox. Moreover, by a result of Labb\'{e} et al.~\cite{LaMaSa98}, this also implies that Stackelberg network pricing games in series-parallel graphs with a single follower are polynomial-time solvable.

We then consider a generalization of Stackelberg network pricing games. We consider the setting in which followers, instead of an  $s$-$t$ path, choose a basis of a given matroid, for instance, a spanning tree in a graph. We prove that in this setting the negative price paradox cannot occur. Lastly, we consider clutters, i.e., set systems that only contain inclusion-wise minimal sets. We derive a necessary and a closely related sufficient condition for clutters to admit the negative price paradox, but a complete characterization remains open.

\subsection{Related literature}
Stackelberg competition was first introduced by Von Stackelberg~\cite{St34} and is now commonly used to describe leader-follower models. Stackelberg network pricing games gained attention due to Labb\'{e} et al.~\cite{LaMaSa98}, who used the game to model road tolling problems. They showed that the problem of finding the optimal prices for the leader when there is only one follower that chooses a shortest path is NP-hard when prices have lower bounds. Roch et al.~\cite{RoSaMa05} proved the more general result that the problem is also NP-hard when prices are unrestricted. Joret~\cite{Jo11} showed that the problem is even APX-hard and Briest et al.~\cite{BrChKhLaNa10} gave a first approximation threshold, namely, the problem cannot be approximated within a factor of $2-o(1)$. For a more detailed survey on this problem, see van Hoesel~\cite{Ho08}. Recently, also different combinatorial problems were studied in a Stackelberg setting. For example, Cardinal et al.~\cite{CaDeFiJoLaNeWe11} proved that the Stackelberg single-follower minimum spanning tree problem is APX-hard. Bil\`{o} el al.~\cite{BiGuPrWi08} and later Cabello~\cite{Ca12} studied Stackelberg shortest path tree games.

Balcan et al.~\cite{BaBlMa08} and Briest et al.~\cite{BrHoKr12} considered single-price strategies. They both show independently that this very simple pricing strategy provides a logarithmic approximation algorithm. B\"{o}hnlein et al.~\cite{BoKrSc17} extended the analysis of this simple algorithm beyond the combinatorial setting.

One more common application of pricing in road tolling problems is to restore inefficiency. Beckmann et al.~\cite{BeMcWi56} and Dafermos and Sparrow~\cite{DaSp69} showed that marginal tolls induce efficient flows when considering the model with congestion effects introduced by Wardrop~\cite{Wa52}. Optimal tolls also exist when users are heterogeneous with respect to
the trade-off between time and money. See, e.g., \cite{CoDoRo03,FlJaMa04,KaKo04,YaHu04}. Hearn and Ramana~\cite{HeRa98} and later Harks et al.~\cite{HaScSi08} considered the problem of characterizing the set of all optimal tolls and used this polyhedron to propose secondary optimization problems. Recently, Basu et al.~\cite{BaLiNi15} improved some of the complexity results for one of these problems called the minimum tollbooth problem.

Efficiency can also be obtained when tolls are imposed by selfish leaders, i.e., leaders that want to maximize their profit as in Stackelberg pricing games. Acemoglu and Ozdaglar~\cite{AcOz04} considered parallel networks in which the leader owns all edges and followers have a fixed reservation value for travel. Their main result is that the monopolist extracts all surplus from the followers and induces the optimal flow. Huang et al.~\cite{HuOzAc06} showed that this result can be generalized to other topologies. Recently, more attention has been addressed to the problem in which multiple leaders compete for followers and thereby induce inefficient flows. Acemoglu and Ozdaglar~\cite{AcOz07} introduced the model with parallel networks and inelastic users. This was later generalized by Ozdagler~\cite{Oz08} and Hayrapetyan et al.~\cite{HaTaWe07} to elastic users, and by Johari et al.~\cite{JaWeRo10} to firms that have to make entry and investment decisions. Correa et al.~\cite{CoGuLiNiSc18} and Harks et al.~\cite{HaScVe18} considered the extension in which a central authority sets price caps for the leaders so as to minimize inefficiencies.

A seemingly related paradox is Braess's paradox \cite{Br68}. It describes the phenomenon in which the increase of resources, like building a new road in a network, may in fact lead to larger costs for the users. Milchtaich~\cite{Mi06} derived a characterization that shows that for undirected single-commodity networks, series-parallel graphs are the largest class of graphs for which Braess's paradox does not occur. This result has been generalized by Chen et al.~\cite{ChDiHu15} and Cenciarelli et al.~\cite{CeGoSa16} to directed graphs and multi-commodity instances. Roughgarden \cite{Ro06} investigated how to improve the performance of a network when it is allowed to remove edges. Fujishige et al.~\cite{FuGoHaPeZe17} characterized the presence of Braess's paradox for clutters. Namely, matroids are exactly those clutters for which Braess's paradox does not occur.

In the mechanism design literature, the performance of selling optimal bundles is also studied. There, it is well-known that revenue maximization with more than one good is a difficult problem. For some results on the performance of selling optimal bundles, see, e.g., Li and Yao~\cite{LiYa13}, and Hart and Nissan~\cite{HaNi17}, and the references therein.

\section{Model}
A Stackelberg network pricing game is given by a tuple $\mathcal{M}=(G,(c_e)_{e\in E},E_p,K,(s^k,t^k,R^k)_{k\in K})$, where $G=(V,E)$ is a directed multigraph, $c_e\in\mathbb{R}_+$ is the fixed cost of edge $e\in E$, $E_p\subseteq E$ is the set of \textit{priceable edges}, $K=\{1,\ldots,\bar k\}$ is the set of followers, $(s^k, t^k)\in V\times V$ is the source-sink pair and $R^k\in\mathbb{R}_+$ is the reservation value for each $k\in K$. Let $m=|E_p|$ and for each $k\in K$, let $\mathcal{P}^k$ denote the set of $s^k,t^k$-paths.

A Stackelberg network pricing game contains two types of players: one leader, and one or more followers. For each priceable edge $e\in E_p$, the leader can specify a price $p_e\in\mathbb{R}$. Let $p=(p_e)_{e\in E_p}$ denote a vector of prices. Given a vector of prices $p$, the total costs of a path $P\in\mathcal{P}^k$ for follower $k\in K$ are defined by
\[c_P=\sum_{e\in P}c_e+\sum_{e\in P\cap E_p}p_e.\]
For each $p\in\mathbb{R}^{E_p}$, each follower chooses a path $P^k(p)\in\mathcal{P}^k$ with $c_{P^k(p)}\leq R^k$ so as to minimize total costs, and if no such path exists, chooses $P^k(p)=\emptyset$. 
Think of $R^k$ as an unpriceable edge for follower $k\in K$ from source to sink with a fixed cost of $R^k$.  

For each $p\in\mathbb{R}^{E_p}$, the profit of the leader is equal to
\[\pi(p)=\sum_{k\in K}\sum_{e\in P^k(p)\cap E_p} p_e.\]
We assume that the leader chooses $p$ so as to maximize profit.

We make the following two assumptions. First, we assume that if there are ties in the lower level, the follower chooses the most profitable path for the leader. This is to make sure that we do not have to deal with arbitrarily small epsilons. Second, we assume that the graph is irredundant, i.e., each edge is contained in at least one $s^k$-$t^k$ path for some $k\in K$. Edges that are on no such path are not relevant for our problem and can be deleted.

We call a price vector $p\in\mathbb{R}^{E_p}$ \textit{optimal} if for all $p'\in\mathbb{R}^{E_p}$, $\pi(p)\geq\pi(p').$
We denote an optimal solution by $p^*$. We call a price vector $p\in\mathbb{R}_+^{E_p}$ \textit{optimal for non-negative prices} if for all $p'\in\mathbb{R}_+^{E_p}$, $\pi(p)\geq\pi(p').$
We denote an optimal solution for non-negative prices by $p_+^*$. For a given model $\mathcal{M}$, we define the \textit{price of positivity} by
\[PoP(\mathcal{M})=\frac{\pi(p^*)}{\pi(p_+^*)}.\]

Labb\'{e} et al.~\cite{LaMaSa98} already observed that an optimal solution for the leader may involve negative prices. We study the question of how much profit a leader can loose by restricting prices to be positive.

\section{Optimal Profit}

In the following section, we will prove our main result. We characterize the loss in profit due to assuming that prices are non-negative.

\begin{theorem} \label{thm:main}
  Let $\mathcal{F}(m,\bar k)$ be the family of Stackelberg network pricing games with $|E_p|=m$ and
  $|K|=\bar{k}\leq 2^m$. Then $\sup_{\mathcal{M}\in \mathcal{F}(m,\bar k)} PoP(\mathcal{M})=\Theta\left(\log(m\cdot\bar k)\right)$.
  \label{th:main_theorem}
\end{theorem}


As a warmup, we give an example that illustrates the phenomenon we want to study. A similar example was first given by Labb\'{e} et al.~\cite{LaMaSa98}.
\begin{example}
  \label{ex:BraessGraph}
Consider a game with one follower that chooses a shortest path from $s$ to $t$ in the network of Figure \ref{fig:whe}.
The priceable edges are depicted
as thicker arrows, and right above of each edge is its cost or price. The reservation value is $R=3$.
Let $P_1, P_2$ and $P_3$ denote the paths defined by the sequences of nodes $(s,u,v,t)$, $(s,u,t)$ and  $(s,v,t)$, respectively.
\begin{figure}[h]
\centering
\begin{tikzpicture}[scale=.6,->,shorten >=1pt,auto,node distance=2.5cm,
  thick,main node/.style={circle,fill=blue!20,draw,minimum size=20pt,font=\sffamily\Large\bfseries},source node/.style={circle,fill=green!20,draw,minimum size=15pt,font=\sffamily\Large\bfseries},dest node/.style={circle,fill=red!20,draw,minimum size=15pt,font=\sffamily\Large\bfseries}]
\node[source node] (source) at (-4.5,0) {$s$};
\node[main node] (u) at (-1.5,0) {$u$};
\node[main node] (v) at (1.5,0) {$v$};
\node[dest node] (sink) at (4.5,0) {$t$};
\draw (source) to [bend left=45] node[above] {$1$} (v);
\draw[ultra thick] (source) to node[above] {$p_1$} (u);
\draw[ultra thick] (u) to node[above] {$p_2$}  (v);
\draw (u) to [bend right=45] node[above] {$1$} (sink);
\draw[ultra thick] (v) to node[above] {$p_3$}  (sink);
\end{tikzpicture}
\caption{A Braess graph} \label{fig:whe}
\end{figure}
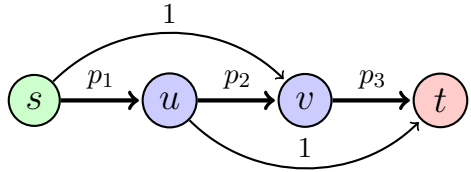

Suppose that $p_1,p_2,p_3\geq 0$.  If the leader wants to induce $P_1$ as a shortest path for the follower, then 
necessarily $p_1+p_2\leq 1$ and $p_2+p_3\leq 1$ and thus $\pi(p)=p_1+p_2 + p_3\leq 2$. If the leader wants the follower to choose path $P_2$, then $p_1+1\leq 3$ and thus $\pi(p)=p_1\leq 2$. Similarly for path $P_3$. Combining these three statements implies $\pi(p_+^*)=2$.

Now, suppose that $p^*=(p_1,p_2,p_3 )=(3,-3,3)$, then $P_1$ has a cost of 3, and $P_2$ and $P_3$ have a cost of 4, and thus the follower will choose path $P_1$. Hence $\pi(p^*)=3$ and $PoP(\mathcal{M})=3/2$.

Negative prices can be seen as a simplified type of bundling scheme. Consider the following situation that can be modeled
by the same network.
A company is selling high quality pasta
and high quality tomato sauce to buyers that want a high quality meal for $\$3$ (a meal is of high quality if the pasta or the sauce has high quality). A low quality
pasta is being sold for $\$1$ and a low quality tomato sauce for $\$1$. If the company sets a positive
price for each individual product, the maximum possible profit is $\$2$. But it might as well give a discount of $\$1$
if both high quality products are bought together, getting the full profit of $\$3$.
\end{example}

We now prove the main theorem. We separate it into three distinct theorems: one for the upper bound, one for a lower bound with $\bar{k}=1$ and another for a lower bound with $\bar{k}=2^m-1$.

Let $\ell^k_v(p)$ denote the cost of a shortest $s-v$
path of follower $k\in K$ for a given vector of prices $p\in\mathbb{R}^{|E_p|}$.

\begin{theorem}\label{thm:upm}
Let $\mathcal{M}\in \mathcal{F}(m,\bar k)$ and $H_n=\sum_{i=1}^n\frac{1}{i}$ denote the $n$-th harmonic number. Then
\[PoP(\mathcal{M})\leq H_{m \bar k}.\]
\end{theorem}
\begin{proof}
Let $p^*\in\mathbb{R}^{E_p},p_+^*\in\mathbb{R}_+^{E_p}$ denote an optimal, optimal positive solution of the leader, respectively. Notice that 
\[\pi(p^*)\leq \sum_{k\in K}R^k-\ell^k_t(0).\]
Briest et al.~\cite{BrHoKr12} proved that for all $\epsilon>0$ there exists an algorithm that in polynomial time calculates a  single-price strategy, i.e., one that sets the same price on all priceable edges, that yields a profit of at least 
\[
\frac{\sum_{k\in K}R^k-\ell^k_t(0)}{(1+\epsilon)\cdot H_{m \bar k}}
\]
and thus
\[\pi(p_+^*)\geq \frac{\sum_{k\in K}R^k-\ell^k_t(0)}{H_{m \bar k}}.\]
Combining the above two statements yields
\[\frac{\pi(p^*)}{\pi(p_+^*)}\leq H_{m \bar k}.\]
\end{proof}

The next result shows that the price of positivity can be arbitrarily large, even if there is only a single follower.
\begin{theorem}\label{thm:inf}
For all $n\in\mathbb{N}$, there exists a model $\mathcal{M}$ with $m=2^{2n-1}-1$ and $\bar{k}=1$
such that
\[PoP(\mathcal{M})=n.\]
\end{theorem}
\begin{proof}
Assume there is only one follower, i.e., $\bar k=1$, that has reservation value $R$. For ease of notation, we omit the superscripts that denote the followers. Consider the network of Figure \ref{fig:bra}, where $h\geq1$. Besides $s$ and $t$, there are $h$ left nodes $v^\ell_1,v^\ell_2,\dots,v^\ell_h$ and $h$ right nodes $v^r_1,v^r_2,\dots,v^r_h$. There is a fixed cost edge
from $s$ to each left vertex $v^\ell_i$, with cost $c^\ell_i$, for each $i\in\left\{ 1,\dots,h \right\}$.
There is a fixed cost edge from each right vertex $v^r_i$ to $t$, with cost $c^r_i$, for each $i\in\left\{ 1,\dots,h \right\}$. There is a priceable edge from $v^\ell_i$ to $v^r_i$ for each $i\in\left\{ 1,\dots,h \right\}$ and a priceable edge from $v^r_i$ to $v^\ell_{i+1}$ for each $i\in\left\{ 1,\dots,h-1 \right\}$. Denote their prices as $p^r_i$ and $p^\ell_i$, respectively.
Assume that $c^{\ell}_1=c^r_h=0$, and $c^{\ell}_i, c^r_i\geq 0$ for all $i\in\{1,\ldots,h\}$. 

\begin{figure}[h]
\centering
\begin{tikzpicture}[scale=.9,->,shorten >=1pt,auto,node distance=2.5cm,
  thick,main node/.style={circle,fill=blue!20,draw,minimum size=20pt,font=\sffamily\Large\bfseries},source node/.style={circle,fill=green!20,draw,minimum size=15pt,font=\sffamily\Large\bfseries},dest node/.style={circle,fill=red!20,draw,minimum size=15pt,font=\sffamily\Large\bfseries}]
\node[source node] (source) at (-3,0) {$s$};
\node[main node] (upl) at (0,2.8) {};
\node[main node] (upr) at (3,2.8) {};
\node[main node] (midupl) at (0,1) {};
\node[main node] (midupr) at (3,1) {};
\node[main node] (middownl) at (0,-1) {};
\node[main node] (middownr) at (3,-1) {};
\node[main node] (downl) at (0,-3) {};
\node[main node] (downr) at (3,-3) {};
\node[dest node] (sink) at (6,0) {$t$};
\node[draw=none] (inv1) at (1.5,2) {\vdots};
\node[draw=none] (inv2) at (1.5,-2) {\vdots};
\draw (source) to node[above] {$c_1^{\ell}$} (upl);
\draw (source) to node[above] {$c_{h/2}^{\ell}$} (midupl);
\draw (source) to node[above] {$c_{h/2+1}^{\ell}$} (middownl);
\draw (source) to node[above] {$c_h^{\ell}$}  (downl);
\draw[ultra thick] (upl) to node[above] {$p^r_1$}  (upr);
\draw (upr) to node[above] {$c_1^r$}  (sink);
\draw[ultra thick,-] (upr) to (inv1);
\draw[ultra thick] (inv1) to (midupl);
\draw[ultra thick] (midupl) to node[above] {$p^r_{h/2}$}  (midupr);
\draw (midupr) to node[above] {$c_{h/2}^r$}  (sink);
\draw[ultra thick] (midupr) to node[above] {$p^\ell_{h/2}$}  (middownl);
\draw[ultra thick] (middownl) to node[above] {$p^r_{h/2+1}$}  (middownr);
\draw (middownr) to node[above] {$c_{h/2+1}^r$}  (sink);
\draw[ultra thick,-] (middownr) to (inv2);
\draw[ultra thick] (inv2) to (downl);
\draw[ultra thick] (downl) to node[above] {$p^r_{h}$}  (downr);
\draw (downr) to node[above] {$c_h^r$}  (sink);
\end{tikzpicture}
\caption{The $h$-th Braess graph.} \label{fig:bra}
\end{figure}
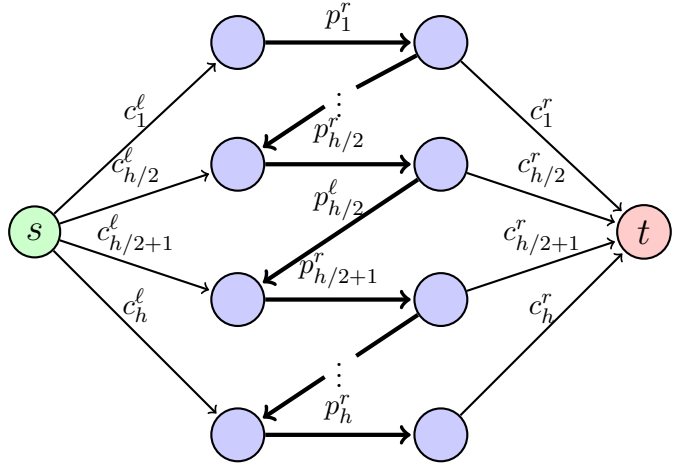

Note first that every
path in the graph is completely defined by an edge leaving $s$ and an edge  entering $t$. 
For  $1\leq i\leq j\leq h$, we define $P_{ij}$ as the path starting with edge $(s,v_i^{\ell})$ and ending with edge $(v_j^r,t)$.
Let $\pi_+^*(P_{ij})$ denote the maximum profit when we impose that path $P_{ij}$ is shortest path and the prices are all
non-negative. Observe that the leader can always delete a priceable edge by setting its price to be $+\infty$.
Notice that path $P_{ij}$ might never be a shortest path, so then we define $\pi_+^*(P_{ij})=0$. 
In order to bound the optimal profit with positive prices, we need the following two lemmas. We bound 
$\pi_+^*(P_{ij})$ first for the case $i=j$ and then for the case $i<j$.

\begin{lemma}\label{lem:dir}
Let $1\leq i\leq h$. If $c_i^{\ell}+c_i^r>R$, then $\pi_+^*(P_{ii})=0$. If $c_i^{\ell}+c_i^r\leq R$, then \[\pi_+^*(P_{ii})\leq R-c_i^{\ell}-c_i^r.\]
\end{lemma}
\begin{proof}
  The fixed cost of path $P_{ii}$ equals $c_i^{\ell}+c_i^r$. If $c_i^{\ell}+c_i^r>R$, then $P_{ii}$ is never a shortest path when prices are non-negative and thus $\pi_+^*(P_{ii})=0$. Otherwise, the cost of $P_{ii}$ for the
  follower is $p^r_i + c^\ell_i + c^r_i$, so if it is less than the reservation value $R$, we must have that $\pi_+^*(P_{ii})
  = p^r_i\leq R-c_i^{\ell}-c_i^r$.
\end{proof}

\begin{lemma}\label{lem:zig}
Let $1\leq i\leq j'< i'\leq j\leq h$. If $c_i^{\ell}>c_{i'}^{\ell}$, then $\pi_+^*(P_{ij})=0$. If $c_j^r>c_{j'}^r$, then $\pi_+^*(P_{ij})=0$. If $c_i^{\ell}\leq c_{i'}^{\ell}$ and $c_j^r\leq c_{j'}^r$, then
\[\pi_+^*(P_{ij})\leq c_{i'}^{\ell}-c_i^{\ell}+c_{j'}^r-c_j^r.\]
\end{lemma}
\begin{proof}
If $c_i^{\ell}>c_{i'}^{\ell}$, then $P_{i'j}$ is always cheaper for the follower than $P_{ij}$ and thus $P_{ij}$ is never a shortest path.
If $c_j^r>c_{j'}^r$, then $P_{ij'}$ is always cheaper than $P_{ij}$ and thus $P_{ij}$ is never a shortest path.

If $c_i^{\ell}\leq c_{i'}^{\ell}$ and $c_j^r\leq c_{j'}^r$, then we have to make sure that the total cost of $P_{ij}$ for the follower is at most the cost of $P_{i'j}$ and $P_{ij'}$. Then we need to ensure that it is cheaper to reach $v^\ell_{i'}$ following $P_{ij}$ that taking directly the edge $(s,v^\ell_{i'})$, and that it is cheaper to go from $v^r_{j'}$ to $t$ following $P_{ij}$ than taking directly the edge $(v^r_{j'},t)$. This means that $c^\ell_i+\sum_{m=i}^{i'-1} p^r_m + \sum_{m=i}^{i'-1} p^\ell_m \leq c_{i'}^{\ell}$ and $c^r_j+ \sum_{m=j'+1}^{j} p^r_m + \sum_{m=j'}^{j-1} p^\ell_m\leq c_{j'}^r$, and thus $\pi_+^*(P_{ij}) = \sum_{m=i}^j p^r_m + \sum_{m=i}^{j-1}p^\ell_m \leq c_{i'}^{\ell}-c_i^{\ell}+c_{j'}^r-c_j^r$.
\end{proof}

Let $n\in\mathbb{N}$, with $n\geq 2$, and $h=4^{n-1}$. We define fixed costs for the $h$-th Braess graph such
that the maximum profit using non-negative prices is at most $2$ while the maximum profit when allowing negative
prices is $2n$. Let $R=2n$. We will set costs that satisfy  
\begin{align}
  c_i^{\ell}=2n-2-c_i^r=c_{h/2+1-i}^r=2n-2-c_{h/2+1-i}^{\ell}  \; \text{ for all }1\leq i\leq h/2.
  \label{eq:symmetry_h_braess}
\end{align} 
Thus, we only need to specify   $c_{h/2+i}^{\ell}$ for $1\leq i\leq h/2$.
We start by letting $c_{h/2+1}^{\ell}=1$. We define $2n-3$ sets of edges recursively, where each set $i=1,\ldots,2n-3$ consists of $2^{i-1}$ edges. For each $i=1,\ldots,2n-3$, define the fixed costs of the next $2^{i-1}$ edges as the fixed costs of all previously defined $2^{i-1}$ edges plus $1$. More precisely, for each $i=1,\dots,2n-3$, we take
$c^\ell_{h/2+2^{i-1}+j} = c^\ell_{h/2+j} +1$ for all $j=1,\dots,2^{i-1}$. In other words, we define the following sequence:
\[(1);(2);(2,3);(2,3,3,4) ; (2,3,3,4,3,4,4,5);\ldots;(2,3,\dots,2n-2).\]

Now, we will prove that given these fixed costs, $\pi_+^*(P_{ij})\leq 2$ for all $1\leq i\leq j\leq h$.
Firstly, by condition (\ref{eq:symmetry_h_braess}), $c_i^{\ell}+c_i^r=2n-2$ for all $1\leq i\leq h$, and thus by Lemma \ref{lem:dir}, $\pi_+^*(P_{ii})\leq 2$. 

Secondly, for each path $P_{ij}$ with $i\leq h/2$ and $j\geq h/2+1$, Lemma \ref{lem:zig} with $i'=h/2+1$ and $j'=h/2$ implies $\pi_+^*(P_{ij})\leq 2-c_i^{\ell}-c_j^r\leq 2$. Then by symmetry of the constructed graph, we can restrict ourselves to paths $P_{ij}$ with $i\geq h/2+1$.

Thirdly, consider a path $P_{ij}$ with $i\geq h/2+1$ and suppose that
\begin{align}
  \frac{h}{2}+1\leq i\leq \frac{h}{2} +2^{i''-1} \text{ and } \frac{h}{2} +2^{i''-1}+1 \leq j \leq \frac{h}{2} + 2^{i''}
  \text{ for some } 1\leq i''\leq 2n-3.
  \label{eq:power_condition}
\end{align}
 Necessarily $c^\ell_i\geq 1$ and $c^r_j = 2n -2 - c^\ell_j \geq 2n -2- (i''+1)$ from
the definition of the costs. If we take $i'=h/2+2^{i''-1}+1$ and $j'=h/2+2^{i''-1}$, then
$c^\ell_{i'}= 2$ and $c^r_{j'}= 2n-2-i''$. So Lemma \ref{lem:zig} 
implies that $\pi_+^*(P_{ij})\leq c^\ell_{i'}-c^\ell_{i}+ c^r_{j'} -c^r_{j} \leq  2$.
If condition (\ref{eq:power_condition}) does not hold, then $h/2 + 2^{i''-1}+1
\leq i\leq j \leq h/2 + 2^{i''}$, for some $1\leq i''\leq 2n-3$. Therefore, by lemma
\ref{lem:zig}, for any $i\leq j'<i'\leq j$, we have that
$\pi_+^*(P_{ij})\leq c^\ell_{i'}-c^\ell_{i}+ c^r_{j'} -c^r_{j}$. But from the definition of the costs,
$c^\ell_{i'}-c^\ell_{i}+ c^r_{j'} -c^r_{j} = c^\ell_{i'- 2^{i''-1}}-c^\ell_{i-2^{i''-1}}+ c^r_{j'-2^{i''-1}} -c^r_{j-
  2^{i''-1}}$. So any bound that we can derive by applying Lemma \ref{lem:zig} for $P_{i-2^{i''-1},j-2^{i''-1}}$, also
  holds for $P_{ij}$. If $i<j$ we can iterate this procedure until condition (\ref{eq:power_condition}) holds, so
  we conclude that $\pi^*_+(P_{ij})\leq 2$.

Combining the above three steps yields $\pi_+^*(P_{ij})\leq 2$ for all $1\leq i\leq j\leq h$. Notice that with the prices
$p^r_i = R$ for $i=1,\dots,h$ and $p^\ell_i=-R$ for $i=1,\dots,h-1$, every $s,t$-path costs at least $R$ and path
$P_{0h}$ yields a profit of $R$. Thus, $PoP(\mathcal{M})=\frac{2n}{2}=n$ for each $n\in\mathbb{N}$.


\end{proof}

\begin{remark}
The proof of Theorem \ref{thm:inf} can also be used to prove the following statement. For each $n\in N$, there exists a model $\mathcal{M}$, with an optimal price vector $p^*$ and an optimal positive prices vector $p^*_+$ such that the cost for the follower is $n$ times smaller under $p^*_+$ than under $p^*$.
\end{remark}

The last result of this section models the flight carrier situation discussed in the introduction. It proves that the price of positivity can be arbitrarily large for series-parallel networks, as long as there are sufficiently many followers.

\begin{theorem}\label{thm:mul}
For all $m\in\mathbb{N}$, there exists a model
$\mathcal{M}$ consisting of a path graph $G$ with $m$ priceable edges and $\bar k =2^m-1$ followers such that
\[PoP(\mathcal{M})= m/2.\]
\end{theorem}

\begin{proof}
  Consider the network of Figure \ref{fig:ser}, where $m\geq 1$. Let $c_i=2^{m-i+1}$ for $i=1,\ldots,m$. There are $m$ groups of followers. For $i=1,\dots,m$ group $i$ has $2^{i-1}$ distinct followers, all with  source-sink pair $(s,t^i)$ and reservation value $c_i$.
\begin{figure}[h]
\centering
\begin{tikzpicture}[scale=.8,->,shorten >=1pt,auto,node distance=2.5cm,
  thick,main node/.style={circle,fill=blue!20,draw,minimum size=20pt,font=\sffamily\Large\bfseries},source node/.style={circle,fill=green!20,draw,minimum size=25pt,font=\sffamily\Large\bfseries},dest node/.style={circle,fill=red!20,draw,minimum size=15pt,font=\sffamily\Large\bfseries}]

  \node[source node,scale=0.9] (1) at (0,0) {$s$};
  \node[dest node,scale=0.9] (2) at (3,0) {$t^1$};
  \node[dest node,scale=0.9] (3) at (6,0){$t^2$};
  \node (4) at (7.5,0) {$\dots$};
  \node[dest node,scale=0.9] (5) at (9,0) {$t^n$};
  
   \draw[ultra thick] (1) to node[above] {$p_1$} (2);
   \draw (1) to [bend right = 30] node[below] {$c_1$} (2);
   \draw (1) to [bend right = 45] node[below] {$c_2$} (3);
   \draw (1) to [bend right = 50] node[below] {$c_n$} (5);
   \draw[ultra thick] (2) to node[above] {$p_2$} (3);
   \draw[-,ultra thick] (3) to node[above] {} (4);
   \draw[ultra thick] (4) to node[above] {$p_n$} (5);
\end{tikzpicture}
\caption{The $m$-th path graph.} \label{fig:ser}
\end{figure}
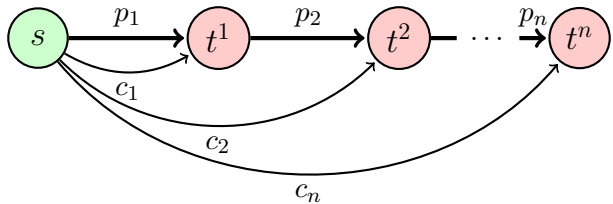

In order to bound the optimal profit with positive prices, we  claim that if a follower from group $i$ travels, then all followers $j=1,\ldots, i$ travel. This follows from the fact that $(c_i)_{i=1}^n$ is a decreasing sequence. So in order to maximize profit, the leader can restrict himself to choosing $p_1\in\{c_1,\ldots,c_m\}$ and $p_i=0$ for all $i=2,\ldots,m$. Thus, taking $p_1=c_i$ gives a profit of $2^{m-i+1} \sum_{j=1}^i 2^{i-1} \leq 2\cdot 2^m$.

The price vector $p^*_i=c_i-c_{i-1}$, where $c_0=0$, for all $i=1,\ldots,m$ yields a profit of $\sum_{i=1}^m 2^m =m\cdot 2^m$. Hence $PoP(\mathcal{M})\geq m/2$.
\end{proof}

\section{Immune Structures to the Negative Price Paradox}
We divide the following section into three subsections. The first subsection considers a single shortest path follower. We characterize the class of topologies that are immune to the negative price paradox. The second subsection considers a different variant of the model in which followers choose the basis of a matroid. We prove that all these models are immune to the negative price paradox. The third and last subsection considers a generalization of the model that includes both of the variants studied. We try to characterize the class of clutters that are immune to the negative price paradox.

\subsection{Shortest path follower}
Assume $\bar k=1$. Let $G_{st}$ denote a graph $G=(V,E)$ with a fixed source-sink pair $(s,t)$. We say that a graph $G_{st}$ is \textit{immune to the negative price paradox} if for all models $\mathcal{M}=(G_{st},(c_e)_{e\in E},E_p,R)$, $PoP(\mathcal{M})=1$. A directed $s$-$t$ graph is \textit{series-parallel} if it either consists of a single edge $(s,t)$, or is obtained from two series-parallel graphs with terminal $(s_1,t_1)$ and $(s_2,t_2)$ composed either in series or in parallel. In a \textit{series composition}, $t_1$ is identified with $s_2$, $s_1$ becomes $s$ and $t_2$ becomes $t$. In a \textit{parallel composition}, $s_1$ is identified with $s_2$ and becomes $s$, and $t_1$ is identified with $t_2$ and becomes $t$.

The main result in this section shows that the leader can extract all surplus from the follower in series-parallel graphs. This result has two interesting consequences. First, it implies that series-parallel graphs are exactly those graphs that are immune to the negative price paradox. Second, it implies that Stackelberg network pricing with a single follower are polynomial-time solvable in series-parallel graphs, as Labb\'{e} et al.~~\cite{LaMaSa98} shows that once the induced path of the follower is known, the problem is polynomial-time solvable.

\begin{theorem}\label{thm:sur}
Let $\bar k=1$. If $G$ is series-parallel, then $\pi(p_+^*)=\min\{\ell_t(+\infty)-\ell_t(0),R-\ell_t(0)\}$.
\end{theorem}
\begin{proof}

We use what is called an open ear decomposition of the graph. An open ear is a directed (simple) path. An open ear decomposition of $G$ is a partition of $E$ into open ears $\{E_1,\ldots, E_{\bar h}\}$ such that the two endpoints of each open ear $E_i$, with $i\geq 2$, belong to some previous open ears $E_j$ and $E_{j'}$, with $j<i$ and $j'<i$, and each internal vertex of each open ear $E_i$ does not belong to $E_j$ for all $j<i$. Given an open ear decomposition $\{E_1,\ldots, E_{\bar h}\}$, we say that $E_i$ is nested in $E_j$ if $i>j$ and the endpoints of $E_i$ both appear in $E_j$. For such $i$ and $j$, let the nest interval of $E_i$ and $E_j$ be the path in $E_j$ between the two endpoints of $E_i$. An open ear decomposition $\{E_1,\ldots, E_{\bar h}\}$ is nested if for each $i\geq 2$, $E_i$ is nested in some $E_j$, and if $E_i$ and $E_{i'}$ are both nested in $E_j$, then either the nest interval of $E_i$ contains that of $E_{i'}$, or vice versa, or the two nest intervals are disjoint. An open ear decomposition can be found starting with an arbitrary $s$-$t$ path and greedily adding open ears. Eppstein~\cite{Ep92} showed that the decompositions found this way are always nested in series-parallel graphs. Now starting with the optimal path $P(0)$, we describe a procedure to construct a specific ear decomposition $\{P(0),E_2,\ldots, E_{\bar h}\}$, which, since $G$ is series-parallel, must be nested.

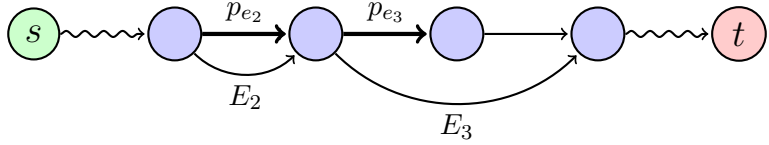
\begin{figure}[h]
\centering
\begin{tikzpicture}[->,shorten >=1pt,auto,node distance=2cm,
  thick,main node/.style={circle,fill=blue!20,draw,minimum size=20pt,font=\sffamily\Large\bfseries},source node/.style={circle,fill=green!20,draw,minimum size=15pt,font=\sffamily\Large\bfseries},dest node/.style={circle,fill=red!20,draw,minimum size=15pt,font=\sffamily\Large\bfseries},scale=0.75]
\node[source node] (ssource) at (-6.25,0) {$s$};
\node[main node] (source) at (-3.75,0) {};
\node[main node] (up) at (-1.25,0) {};
\node[main node] (down) at (1.25,0) {};
\node[main node] (sink) at (3.75,0) {};
\node[dest node] (ssink) at (6.25,0) {$t$};
\draw[->,decorate,decoration={snake,amplitude=.4mm,segment length=2.5mm,post length=1mm}] (ssource) to (source);
\draw (source) to (up);
\draw[ultra thick] (source) to node[above] {$p_{e_2}$}  (up);
\draw (source) to [bend right=45] node[below] {$E_2$} (up);
\draw (up) to [bend right=45] node[below] {$E_3$}  (sink);
\draw[ultra thick] (up) to node[above] {$p_{e_3}$} (down);
\draw (down) to (sink);
\draw[->,decorate,decoration={snake,amplitude=.4mm,segment length=2.5mm,post length=1mm}] (sink) to (ssink);
\end{tikzpicture}
\caption{A series-parallel graph is immune to the negative price paradox.} \label{fig:sp}
\end{figure}

In this proof, we consider the graph $G'$, which is defined as $G$ with an unpriceable edge from $s$ to $t$ with a fixed cost of $R$. Notice that the problems in $G$ and $G'$ are esssentially equivalent. In particular, $G'$ is series-parallel, $\ell'_t(+\infty)=\min\{\ell_t(+\infty),R\}$, and $\pi(p_+^*)$ in $G'$ equals $\pi(p_+^*)$ in $G$. Let $P(0)$ be a shortest $s$-$t$ path in $G'$ when the prices are all set to $0$. Denote by $p'$ the vector of non-negative prices that yields the maximum profit under the restriction that $P(p')=P(0)$. 

Assume that $\pi(p')<\ell'_t(+\infty)-\ell_t(0)$. We prove that there exists a path $\hat{P}\subseteq E'\setminus E_p$ with cost strictly less than $\ell'_t(+\infty)$, which contradicts the definition of $\ell'_t(+\infty)$. Firstly, create the vector $p''$ starting with $p'$ and setting the price of all priceable edges
outside $P(0)$ to $+\infty$. Clearly $\pi(p'')=\pi(p')$. Take any $e_2 \in P(0)\cap E_p$. Note that there must be an open ear $E_2$ (a detour for $e_2$) such that the fixed costs of $E_2$ are equal the total costs of the nest interval $\hat{E}_2\subseteq P(0)$. If such an open ear would not exist, the leader could marginally increase the price $p_{e_2}''$ without changing the shortest path of the follower, yielding a strictly higher profit. Now take $e_3\in P(0)\cap E_p \setminus \hat{E}_2$.
Analogously, there must be an open ear $E_3$ with nesting interval $\hat{E}_3\subseteq P(0)$ such that
the total cost of $E_3$ is equal to the total cost of $\hat{E}_3$. Inductively, define
$e_h\in P(0)\cap E_p\setminus \bigcup_{i=1}^{h-1} \hat{E}_i$, an open ear $E_h$ that is a detour for $e_h$,
and its nesting interval $\hat{E}_h\subseteq P(0)$ such that the cost of $E_h$ is equal to the cost of $\hat{E}_h$, until 
$P(0)\cap E_p\setminus \bigcup_{i=1}^{h-1} \hat{E}_i =\emptyset$.
Since $G'$ is series-parallel, the open ear decomposition must be nested, so using ears with disjoint nesting intervals
we can define an $s$-$t$ path $\hat{P}$ with the same cost as $P(0)$ under $p''$ that uses only fixed cost edges. We conclude noting
that the cost of $P(0)$ under $p''$ is $\pi(p') + \ell'_t(0)<\ell'_t(+\infty)$.

\end{proof}

\begin{corollary}\label{thm:sp}
Let $\bar k=1$. A graph $G_{st}$ is immune to the negative price paradox if and only if $G_{st}$ is a series-parallel graph.
\label{th:SP_singleF}
\end{corollary}
\begin{proof}
First, assume that $G$ is a series-parallel graph. Recall that $\pi(p^*)\leq R-\ell_t(0)$ and $\pi(p^*)\leq \ell_t(+\infty)-\ell_t(0)$, and thus the result follows from Theorem \ref{thm:sur}.

Second, assume that $G_{st}$ is not a series-parallel graph. We show that there is a model $\mathcal{M}=(G_{st},(c_e)_{e\in E},E_p,R)$ such that $PoP(\mathcal{M})>1$. Basically, if the graph is not series parallel one can find a subgraph
with the same structure as Example~\ref{ex:BraessGraph}.

We call a subgraph $G'$ of $G$ an $s$-$t$ paradox if $G'=P_1\cup P_2\cup P_3$ is the union of three paths $P_1,P_2,P_3$ with the following properties:
\begin{enumerate}
  \item[(i)] $P_1$ is an $s$-$t$ path going through distinct vertices $a,u,v,b$ such that $s\preceq_{P_1} a\prec_{P_1} u\prec_{P_1} v\prec_{P_1} b\preceq_{P_1} t$, where $(\prec_{P_1})$ denotes the order in which $P_1$ visits the nodes.
\item[(ii)] $P_2$ is an $a-v$ path with $V(P_2)\cap V(P_1)=\{a,v\}$.
\item[(iii)] $P_3$ is a $u-b$ path with $V(P_3)\cap V(P_1)=\{u,b\}$ and $V(P_3)\cap V(P_2)=\emptyset$.
\end{enumerate}
Chen et al.~\cite{ChDiHu15} prove that if $G_{st}$ is not series-parallel, then $G_{st}$ contains an $s$-$t$ paradox. Let $G_{st}'=(V',E')=P_1\cup P_2\cup P_3$ be an $s$-$t$ paradox contained in $G_{st}$. Let $e_1,e_2$ be the two outgoing edges from $a$ with $e_1\in P_2$ and $e_2\in P_1$, let $e_3$ be an outgoing edge from $u$ with $e_3\in P_1$, and let $e_4,e_5$ be the two incoming edges to $b$ with $e_4\in P_1$ and $e_5\in P_3$. Define $\mathcal{M}$ as follows: $c_e=1$ if $e\in\{e_1,e_5\}$, $c_e=0$ if $e\in E'\setminus\{e_1,e_5\}$ and $c_e=+\infty$ if $e\in E\setminus E'$, $E_p=\{e_2,e_3,e_4\}$, and $R=3$. See Figure \ref{fig:whe2} for an illustration.

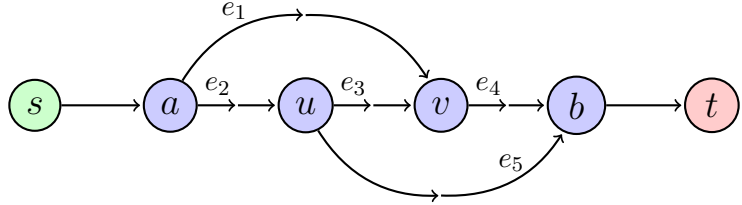
\begin{figure}[h]
\centering
\begin{tikzpicture}[scale=.6,->,shorten >=1pt,auto,node distance=2.5cm,
  thick,main node/.style={circle,fill=blue!20,draw,minimum size=20pt,font=\sffamily\Large\bfseries},source node/.style={circle,fill=green!20,draw,minimum size=15pt,font=\sffamily\Large\bfseries},dest node/.style={circle,fill=red!20,draw,minimum size=15pt,font=\sffamily\Large\bfseries}]
\node[source node] (source) at (-7.5,0) {$s$};
\node[main node] (a) at (-4.5,0) {$a$};
\node[main node] (u) at (-1.5,0) {$u$};
\node[main node] (v) at (1.5,0) {$v$};
\node[main node] (b) at (4.5,0) {$b$};
\node[dest node] (sink) at (7.5,0) {$t$};
\draw (source) to (a);
\draw (a) to [bend left=30] node[above] {$e_1$} (-1.5,2);
\draw (-1.5,2) to [bend left=30] (v);
\draw (a) to node[above] {$e_2$}  (-3,0);
\draw (-3,0) to (u);
\draw (u) to node[above] {$e_3$}  (0,0);
\draw (0,0) to (v);
\draw (u) to [bend right=30] (1.5,-2);
\draw (1.5,-2) to [bend right=30] node[above] {$e_5$} (b);
\draw (v) to node[above] {$e_4$}  (3,0);
\draw (3,0) to (b);
\draw (b) to (sink);
\end{tikzpicture}
\caption{An $s$-$t$ paradox.} \label{fig:whe2}
\end{figure}
Suppose that $p\in\mathbb{R}^{E_p}_+$, then $P(p)$ is either $P_1$, $P_2$ or $P_3$. However, in all three cases, $\pi(p)\leq 2$. Now, suppose that $p^*=(3,-3,3)$, then $\pi(p^*)=3$. Hence $PoP(\mathcal{M})=\frac{3}{2}>1$.
\end{proof}

\begin{corollary}
Stackelberg network pricing games in series-parallel graphs with a single follower are polynomial-time solvable.
\end{corollary}

We would like to stress that Theorem \ref{thm:sur} is only valid for a single follower as the following example shows.

\begin{example}
Consider the network of Figure \ref{fig:sepa}. Each follower $k=1,2$ is defined by $(s^k,t)$, and $R^1=2$ and $R^2=1$.
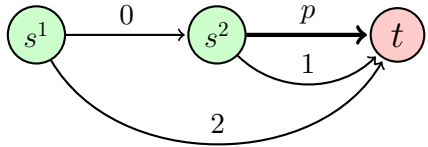
\begin{figure}[h]
\centering
\begin{tikzpicture}[scale=.6,->,shorten >=1pt,auto,node distance=2.5cm,
  thick,main node/.style={circle,fill=blue!20,draw,minimum size=20pt,font=\sffamily\Large\bfseries},source node/.style={circle,fill=green!20,draw,minimum size=15pt,font=\sffamily\Large\bfseries},dest node/.style={circle,fill=red!20,draw,minimum size=15pt,font=\sffamily\Large\bfseries}]
\node[source node,scale=0.8] (source) at (-8,0) {$s^1$};
\node[source node,scale=0.8] (a) at (-4,0) {$s^2$};
\node[dest node] (sink) at (0,0) {$t$};
\draw (source) to node[above] {$0$} (a);
\draw (source) to [bend right=60] node[above] {$2$} (sink);
\draw[ultra thick] (a) to node[above] {$p$} (sink);
\draw (a) to [bend right=45] node[above] {$1$} (sink);
\end{tikzpicture}
\caption{A series-parallel graph.} \label{fig:sepa}
\end{figure}

Clearly, it is impossible for the leader the extract 2 from follower 1 and 1 from follower 2.
\end{example}

Our last result of this subsection shows that the leader can extract all surplus from the follower by using negative prices if the shortest path with respect to fixed price edges is owned by the leader. In particular, these problems are polynomial solvable.

\begin{lemma}\label{lem:neg}
Let $\bar k=1$ and $P(0)$ denote a shortest $s$-$t$ path when the prices are all set to $0$. If $P(0)\subseteq E_p$, then
$\pi(p^*)=\min\{\ell_t(+\infty)-\ell_t(0),R-\ell_t(0)\}$.
\end{lemma}
\begin{proof}
We will show that the leader can extract all surplus from the follower. Notice that we can assume that $p_e=+\infty$ for all $e\in E_P\setminus P(0)$. The profit maximization problem of the leader is
 \[
\begin{array}{rcccl}
\max	\limits_{(p_e)_{e\in E_p},(\ell_v(p))_{v\in V}}		   & \sum_{e\in P(0)}p_e & & &\\
			   & \ell_v(p)-\ell_u(p)-p_e &=& c_e &\forall e=(u,v)\in P(0)\\
			   & \ell_v(p)-\ell_u(p) &\leq& c_e&\forall e=(u,v)\in E\setminus E_p\\
			   & p_e &=& +\infty&\forall e\in E_p\setminus P(0)\\
			   & \ell_t(p) &\leq& R.&
			   \end{array}
\]
The constraints guarantee that $P(0)$ is a shortest path given price vector $p$.

Define $\delta^-(v)=\{(u,v)\in E\text{ for some }u\in V\}$ and $\delta^+(v)=\{(v,w)\in E\text{ for some }w\in V\}$. The dual of the above linear program is
\[
\begin{array}{rcccl}
\min\limits_{(y_e)_{e\in E},y_{(s,t)}} & \sum_{e\in E} y_e\cdot c_e+y_{(s,t)}\cdot R&&& \\
			   & \sum_{e\in \delta^-(v)}y_e - \sum_{e\in \delta^+(v)}y_e &=&0& \forall v\in V\setminus\{t\}\\
			   & \sum_{e\in \delta^-(t)}y_e+y_{(s,t)}- \sum_{e\in \delta^+(t)}y_e&=&0&\\
			   & y_e &=& -1 & \forall e\in P(0)\\
			   & y_e&\geq& 0 & \forall e\in E\setminus E_p\\
			   & y_e &=& 0 & \forall e\in E_p\setminus P(0)\\
			   & y_{(s,t)}&\geq& 0, &
\end{array}
\]
which is equivalent to
\[
\begin{array}{rcccl}
\min\limits_{(y_e)_{e\in E\setminus E_p},y_{(s,t)}} & \sum_{e\in E\setminus E_p} y_e\cdot c_e+y_{(s,t)}\cdot R-\sum_{e\in P(0)} c_e &&& \\
			   & \sum_{e\in \delta^-(v)\cap (E\setminus E_p)}y_e - \sum_{e\in \delta^+(v)\cap (E\setminus E_p)}y_e &=&0 & \forall v\in V\setminus\{s,t\}\\
			   & \sum_{e\in \delta^-(s)\cap (E\setminus E_p)}y_e - \sum_{e\in \delta^+(s)\cap (E\setminus E_p)}y_e &=&-1\\
			   & \sum_{e\in \delta^-(t)\cap (E\setminus E_p)}y_e+y_{(s,t)} - \sum_{e\in \delta^+(t)\cap (E\setminus E_p)}y_e&=&1\\
			   & y_e&\geq& 0 & \forall e\in E\setminus E_p\\
			   & y_{(s,t)}&\geq& 0.&
\end{array}
\]
The optimal solution of the second dual either picks a shortest path in the graph with edge set $E\setminus E_p$, or sets $y_e=0$ for all $e\in E$ and $y_{(s,t)}=1$. The result then follows by strong duality.
\end{proof}

The result in Lemma \ref{lem:neg} only applies to a single follower.
\begin{example}
Consider the network in Figure \ref{fig:tri}, where $K=\{1,2,3\}$. Let $R^1=1$ and $R^2=R^3=2$. Observe that the leader owns a shortest path from source to sink whenever $p_e=0$ for all $e\in E$ for each follower $k\in K$. 
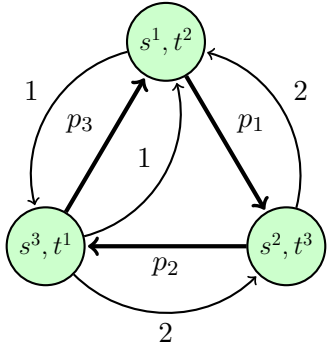
\begin{figure}[h]
\centering
\begin{tikzpicture}[scale=.8,->,shorten >=1pt,auto,node distance=2.5cm,
  thick,main node/.style={circle,fill=blue!20,draw,minimum size=20pt,font=\sffamily\Large\bfseries},source node/.style={circle,fill=green!20,draw,minimum size=25pt,font=\sffamily\Large\bfseries},dest node/.style={circle,fill=red!20,draw,minimum size=15pt,font=\sffamily\Large\bfseries}]

  \node[source node,scale=0.7] (3) at (0,0) {$s^3,t^1$};
  \node[source node,scale=0.7] (2) at (4,0) {$s^2,t^3$};
  \node[source node,scale=0.7] (1) at (2,3.4){$s^1,t^2$};
   \draw[ultra thick] (1) to node[above right] {$p_1$} (2);
   \draw[ultra thick] (2) to node[below] {$p_2$} (3);
   \draw[ultra thick] (3) to node[above left] {$p_3$} (1);
   \draw (1) to [bend right = 45] node[above left] {$1$} (3);
   \draw (2) to [bend right = 45] node[above right] {$2$} (1);
   \draw (3) to [bend right = 45] node[below] {$2$} (2);
   \draw (3) to [bend right = 45] node[above left] {$1$} (1);
\end{tikzpicture}
\caption{A triangle network.}\label{fig:tri}
\end{figure}

The total surplus of all followers equals $\sum_{k\in K}R^k-\ell^k_t(0)=1+2+2=5$. However, the following set is empty and thus the leader cannot extract all surplus:
\[\{(p_1,p_2,p_3)\in\mathbb{R}^3\mid p_1+p_2=1,\;p_2+p_3=2,\;p_1+p_3=2,\;p_3\leq 1\}.\]
\end{example}

\subsection{Matroid followers}

A more general of description of the model we have been looking at is the following. A Stackelberg pricing game is a tuple $\mathcal{M}= \left(E, (c_e)_{e\in E}, E_p, K, (\mathcal{S}^k)_{k\in K},(R^k)_{k\in K}\right)$,
where $E$ is the set of resources, $c_e\in\mathbb{R}$ is a fixed cost for each $e\in E$, $E_p\subseteq E$ is the set of priceable resources, $K=\{1,\ldots,\bar k\}$ is the set of followers, $\mathcal{S}^k\subseteq 2^E$ is the set of strategies and $R^k\in\mathbb{R}_+$ is the reservation value for each $k\in K$.

So far we have studied
the case where $E$ is the set of edges in a network and $\mathcal{S}^k$ is the set of $s^k$-$t^k$ paths for each $k\in K$. Now, we assume that for each $k\in K$ and all $S, S'\in \mathcal{S}^k$ with $S\neq S'$, $S\not\subset S'$. That is, we assume that set system $(E,\mathcal{S}^k)$ is a \textit{clutter} for each $k\in K$.

Given a vector of prices $p\in\mathbb{R}^{E_p}$, the total costs of a set $S\in\mathcal{S}^k$ for follower $k\in K$ are defined by
\[c_S=\sum_{e\in S}c_e+\sum_{e\in S\cap E_p}p_e.\]
For each $p$, each follower chooses a set $S^k(p)\in\mathcal{S}^k$ with $c_{s^k(p)}\leq R^k$ so as to minimize total costs, and if no such set exists, chooses $S^k(p)=\emptyset$. The reservation value $R^k$ for follower $k\in K$ can be thought of as a feasible unpriceable resource with fixed cost $R^k$.

A \textit{matroid} is a tuple $M=(E,\mathcal{I})$, where $E$ is a finite set, called the ground set, and $\mathcal{I}\subseteq 2^E$ is a non-empty family of subsets of $E$, called independent sets, such that: (1) $\emptyset\in\mathcal{I}$, (2) if $X\in\mathcal{I}$ and $Y\subseteq X$, then $Y\in\mathcal{I}$, and (3) if $X, Y\in\mathcal{I}$ with $|X|>|Y|$, then there exists an $e\in X\setminus Y$ such that $Y\cup\{e\}\in \mathcal{I}$ (this property is called the augmentation property). The inclusion-wise maximal independent sets of $\mathcal{I}$ are called bases of matroid $M$. See \cite{Ox92,Sc03,We10} for more information on matroids.

We now consider the case in which instead of a path, each follower $k$
chooses a minimum cost basis of a given matroid $M^k=(E,\mathcal{I}^k)$, which is
a generalization of the setting in which
followers choose spanning trees of a given graph~\cite{CaDeFiJoLaNeWe11}.
We prove that in this case there is no need for negative prices.

\begin{theorem}
	Let $\mathcal{M}$ be a Stackelberg game where $\mathcal{S}^k$ is the
	base set of a matroid $M^k$ for each follower $k\in K$. Then $PoP(\mathcal{M})=1$.
	\label{th:MatroidPoP}
\end{theorem}
We first prove a lemma and then proceed to prove Theorem~\ref{th:MatroidPoP}.
These proofs strongly rely on the fact that matroids are exactly those structures for which the greedy algorithm finds the optimal solution. Greedy algorithms are very simple algorithms in which in each step a local optimal choice is selected. This means that we can assume that for finding a minimum cost base, each follower considers sequentially the next cheapest resource. If adding this resource to the set is feasible, then the resource is selected. If not, the resource is not selected.

\begin{lemma}
	Let $M=(E,\mathcal{I})$ be a matroid with weights $w:E\rightarrow \mathbb{R}$, and
	$A$ a minimum weight base. If $w'$ is obtained by increasing the weight of
	an element $e\in E$, and if $A$ is not optimal under $w'$, then there is an
	element $f\in E$ such that $A-e+f$ is an optimal base under $w'$.
	\label{lem:exchange_matroid}
\end{lemma}

\begin{proof}
	Assume $A$ is not optimal under $w'$, and that we run in parallel the greedy algorithm
	on both instances. Call G1 the greedy algorithm running with weights $w$ and
	G2 the greedy algorithm running with weight $w'$. When looking at the ordered lists, it is clear that
	both lists are equal up to element $e$. So the partial solutions up to element $e$ are the same.
	Now, when G1 finds $e$, it adds it to the solution, but G2 sees
	the next element. Since $A$ is not optimal for $w'$, G2 must find
	an element $f\notin A$ and add it to the partial solution before it reaches
	$e$ in the list induced by $w'$. Besides $e$ and
	$f$, the two algorithms have added the same resources so far. Therefore,
	the two partial solutions have the same size, and then, when one of the
	two algorithms can add a new item, the other can as well, because of the
	augmentation property, whereas if a new item cannot be added for one algorithm, it also cannot be added for the other algorithm, again by the augmentation property. With this observation we can conclude that, at the end,
	the only difference between the two solutions is elements $e$ and $f$.
\end{proof}

\begin{proof}[Proof of Theorem~\ref{th:MatroidPoP}]
	Assume by contradiction that in $p^*$, the optimal solution of the leader,
	there is a negative price $p^*_e<0$	for some resource $e\in E_p$. For each $k\in K$, let $S^k(p^*)$ be the
	optimal strategy for agent
	$k$, given $p^*$, obtained using the greedy algorithm.
	Let now $p'$ be equal to $p^*$ in every resource but $e$, in which
	$p'_e=0$. 
	
	If $S^k(p^*)$ is not optimal under $p'$, by Lemma~\ref{lem:exchange_matroid}
	there is an element $f\in E$ such that $S_k(p^*)-e+f$ is optimal for
	the follower $k$.
	If $f\in E_f$, then the profit for the leader increases because he was
	loosing value in $e$, and if $f\in E_p$, the optimality of $S_k(p^*)$ implies
	that $p^*_e < p^*_f$ and then the leader also increases his profit.
	We conclude that there is a positive solution $p^*_+$ such that
	$\pi(p^*)\leq \pi(p^*_+)$.
\end{proof}

\subsection{Clutters}
Assume that $\bar k=1$. We say that $\left(E, \mathcal{S}\right)$ admits the negative price paradox if there is a model  $\mathcal{M}= \left(E, (c_e)_{e\in E}, E_p, \mathcal{S},R\right)$ such that $PoP(\mathcal{M})>1$. We say that $(E,\mathcal{S})$ is immune to the negative price paradox if there is no model $\mathcal{M}= \left(E, (c_e)_{e\in E}, E_p, \mathcal{S},R\right)$ that admits the negative price paradox. We try to determine the structural
properties of clutters that are immune to the paradox, as we did for the network pricing games. Although we are not
able to give a complete characterization, we show a necessary condition and a closely related sufficient condition.

\begin{proposition} \label{pro:nec}
If $\left(E, \mathcal{S}\right)$ admits the negative price paradox, then there are $a,b,c\in E$ and $A,B,C\in\mathcal{S}$ such that $a\in A\setminus(B\cup C)$, $b\in (A\cap B)\setminus C$, and $c\in (A\cap C)\setminus B$.
\end{proposition}

\begin{proof}
	Assume $\left(E, \mathcal{S}\right)$ admits the negative price paradox. Let $\mathcal{M}= \left(E, (c_e)_{e\in E}, E_p, \mathcal{S},R\right)$ be a model such that $PoP(\mathcal{M})>1$. Let $p^*$ be the optimal solution for the leader, and let $A$ be the optimal
	strategy for the follower, given $p^*$. There must be a priceable resource
	$a\in A$ such that $p^*_a<0$, since $PoP(\mathcal{M})>1$.
	Clearly, there must also be a priceable resource $b\in A$  with
	$p_b^*>0$, because otherwise the profit would be negative.

	Suppose now that the leader increases $p_a$ by a small value $\varepsilon>0$
	and decreases $p_b$ by the same $\varepsilon$. This will preserve the
	total cost of $A$ and then this cannot be done for with every resource with negative
	tolling unless $PoP(\mathcal{M})=1$. So we can assume that we have $a,b$ for which
	this cannot be done for any $\varepsilon>0$ without losing the optimality of $A$.
	Then, there is a strategy $B\in\mathcal{S}$
	such that $a\notin B, b\in B$ and with cost equal to the cost of $A$.

	Now, there must be a third priceable resource $c\in A\setminus B$ as well, such that $p_c>0$. If not, setting $p_a=0$ would increase the profit of the leader by selling $B$, contradicting the optimality for the leader. By the same
	argument as before, there is a strategy $C$ with $a\notin C, c\in C$ and with cost
	equal to the cost of $A$ and $B$. Furthermore, we can assume that $b\notin C$,
	because otherwise we could repeat the argument for a resource $d\in A\setminus C$.
\end{proof}

\begin{proposition} \label{pro:suf}
If there are $a,b,c\in E$ and $A,B,C\in\mathcal{S}$ such that $a\in A\setminus(B\cup C)$, $b\in (A\cap B)\setminus C$, and $c\in (A\cap C)\setminus B$, and there are no $S\in\mathcal{S}$ with $S\subseteq A\cup B\cup C$, then $\left(E, \mathcal{S}\right)$ admits the negative price paradox.
\end{proposition}
\begin{proof}
Assume there are $a,b,c\in E$ and $A,B,C\in\mathcal{S}$ such that $a\in A\setminus(B\cup C)$, $b\in (A\cap B)\setminus C$, and $c\in (A\cap C)\setminus B$, and there are no $S\in\mathcal{S}$ with $S\subseteq A\cup B\cup C$. We define a model $\mathcal{M}= \left(E, (c_e)_{e\in E}, E_p, \mathcal{S},R\right)$ such that $PoP(\mathcal{M})>1$.

Assume that $E_p=\{a,b,c\}$ and $c_e=0$ for all $e\in E_p$. Since $\mathcal{S}$ is a clutter, there exists a $d\in B\setminus A$ and $d'\in C\setminus A$ with $d$ not necessarily different from $d'$. If $(B\cap C)\setminus A\neq\emptyset$, assume $c_d=1$ for some $d\in B\cap C$, and $c_e=0$ for all $e\in (A\cup B\cup C)\setminus\{d\}$. If $(B\cap C)\setminus A=\emptyset$, assume $c_d=1$ for some $d\in B\setminus A$, $c_{d'}=1$ for some $d'\in C\setminus A$, and $c_e=0$ for all $e\in (A\cup B\cup C)\setminus\{d,d'\}$. For all $f\in E\setminus (A\cup B\cup C)$, assume $c_f=+\infty$. Moreover, assume that $R=3$.

First, we consider the maximum profit by assuming that prices are non-negative. The only way the leader can make a profit of $3$ is by selling set $A$. All other sets have a fixed cost of at least $1$. If the leader wants to sell $A$, the costs of sets $B$ and $C$ must be at least $3$ and thus, $p_b\geq 2$ and $p_c\geq 2$. However, this implies that the costs of $A$ are at least $4$ and thus the follower does not buy any set.

Second, we consider the maximum profit by allowing for negative prices. Setting $p_a=-1$ and $p_b=p_c=2$ yields a profit of $3$ by selling set $A$. Notice that sets $B$ and $C$ have a cost of $3$, and all other sets a cost of at least $3$. Hence $PoP(\mathcal{M})>1.$
\end{proof}

\section{Discussion}
We have shown the impact of negative prices on the profit of the leader in Stackelberg network pricing games. The use of negative prices explains, for example, why flight carriers are willing to give discounts on connecting flights. Our main result shows that the loss in profit due to assuming positive prices can be arbitrarily large, even for a single follower. Moreover, we showed that series-parallel graphs with a single follower are exactly those graphs that are immune to this phenomenon.

Several natural questions remain open. One important open question is whether the lower bound on the price of positivity obtained in Theorem \ref{thm:mul} is assymptotically tight in the number of priceable edges. From Briest et al.~\cite{BrHoKr12}, we know that a quadratic upper bound applies, whereas the lower bound only grows linearly. A second open question is whether the necessary condition in Proposition \ref{pro:nec} is also a sufficient condition. We expect that there exists a simple characterization of all structures that are immune to
the negative price paradox, when having a single follower.

A potential generalization for our model is to allow for congestion externalities. A common model for this setting is a non-atomic congestion game with the Wardrop equilibrium \cite{Wa52} as a solution concept. This would also open up research towards models with multiple leaders. The last example shows that when having congestion externalities, the price of positivity can already be larger than one in series-parallel graphs. One immediate question motivated by the examples we have found is how the price of positivity is related to the size of the network when there are congestion externalities. 

\begin{example}
  Consider the network of Figure \ref{fig:se2} and assume $\bar k=1$ with $(s^1,t^1,d^1)=(s,t,1)$. For each $p\in\mathbb{R}_+^2$, assume that the follows choose a Wardrop flow \cite{Wa52} with respect to $c_e(x)+p_e$.
\begin{figure}[H]
\centering
\begin{tikzpicture}[scale=.8,->,shorten >=1pt,auto,node distance=2cm,
  thick,main node/.style={circle,fill=blue!20,draw,minimum size=20pt,font=\sffamily\Large\bfseries},source node/.style={circle,fill=green!20,draw,minimum size=15pt,font=\sffamily\Large\bfseries},dest node/.style={circle,fill=red!20,draw,minimum size=15pt,font=\sffamily\Large\bfseries}]

  \node[source node] (1) at (0,0) {$s$};
  \node[main node] (2) at (3,0) {};
  \node[main node] (3) at (6,0) {};
  \node[dest node] (4) at (9,0){$t$};
   \draw (1) to node[above] {$1/2$} (2);
   \draw[ultra thick] (2) to node[above] {$p_1$} (3);
   \draw (1) to [bend right = 30] node[below] {$x$} (3);
   \draw (1) to [bend right = 45] node[below] {$1$} (4);
   \draw[ultra thick] (3) to node[above] {$p_2$} (4);
\end{tikzpicture}
\caption{A single-commodity series-parallel graph.} \label{fig:se2}
\end{figure}
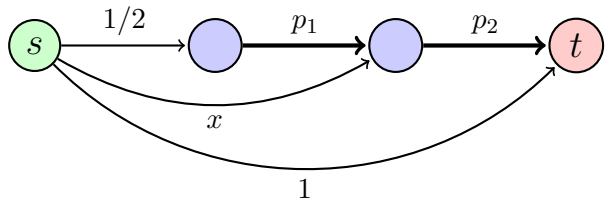

Then the price vector $p^*=(-1/4,3/4)$ yields a profit of $9/16$, whereas $p_+^*=(0,1/2)$ yields a profit of $1/2$.
\end{example}

\section*{Acknowledgments}
We would like to thank the participants of the AGCO-seminar, Jos\'{e} Correa, Britta Peis and Oliver Schaudt for helpful discussions. This work was partially supported by the Millennium Nucleus Information and Coordination in Networks, ICM/FIC RC130003, by the Millennium Institute for
Research in Market Imperfections and Public Policy, MIPP, IS130002, and by the Chilean Science Council under CONICYT-PFCHA/Doctorado Nacional/2018-21180347.

\bibliographystyle{splncs03}

\end{document}